\documentclass[journal]{IEEEtran}

\usepackage{amsmath}
\usepackage{amsfonts}
\usepackage{amssymb}
\usepackage{amsthm}
\usepackage{tabularx}
\usepackage{mathrsfs} 
\usepackage[caption=false,font=footnotesize]{subfig}

\usepackage{url}
\usepackage{hyperref}
\newtheorem{theorem}{Theorem}
\newtheorem{remark}{Remark}
\newtheorem{definition}{Definition}
\newtheorem{Proposition}{Proposition}
\newtheorem{lemma}{Lemma}
\newtheorem{assumption}{Assumption}

\usepackage{stfloats}
\usepackage{float}
\usepackage{graphicx}
\usepackage{cite}
\usepackage{xcolor}

\makeatletter
\def\blfootnote{\xdef\@thefnmark{}\@footnotetext}
\makeatother

\usepackage{amsmath,amssymb,amsfonts}
\usepackage{algorithm}
\usepackage{algorithmic}
\usepackage{graphicx}
\usepackage{cite}
\usepackage{mathtools}
\begin{document}

\title{
	Copula-Induced Correntropy for \\ Robust Conjugate Gradient Learning
}

\author{
	Anonymous Author(s)%
	\thanks{This work will be anonymized for peer review.}
}

\author{Farshad~Rostami~Ghadi,~\IEEEmembership{Member},~\textit{IEEE}, 	F.~Javier~L\'opez-Mart\'inez,~\IEEEmembership{Senior~Member},~\textit{IEEE}, David~Morales-Jim\'enez,~\IEEEmembership{Senior~Member},~\textit{IEEE},~Kai-Kit~Wong,~\IEEEmembership{Fellow},~\textit{IEEE},~and~Marios~Kountouris,~\IEEEmembership{Fellow},~\textit{IEEE}
}
\maketitle
 \blfootnote{The work of F. J. L\'opez-Mart\'inez and D. Morales-Jim\'enez is supported by grant PID2023-149975OB-I00 (COSTUME) funded by MICIU/AEI/10.13039/501100011033, and by ERDF/EU.}
 \blfootnote{The work of F. Rostami Ghadi is supported by the European Union's Horizon Europe Research and Innovation Programme under Marie Sk\l odowska-Curie Grant No. 101107993.}
\blfootnote{The work of K. K. Wong is supported by the Engineering and Physical Sciences Research Council (EPSRC) under Grant EP/W026813/1.}
\blfootnote{The work of M. Kountouris is supported by the European Research Council (ERC) under the European Union’s Horizon 2020 Research and Innovation Programme (Grant agreement No. 101003431).}

\blfootnote{\noindent F. Rostami Ghadi, F. J. L\'opez-Mart\'inez and D. Morales-Jim\'enez are with the Department of Signal Theory, Networking and Communications, Research Centre for Information and Communication Technologies (CITIC-UGR), University of Granada, 18071, Granada, Spain (e-mail: $\rm \{f.rostami, fjlm,dmorales\}@ugr.es$).}
\blfootnote{\noindent 
K. K. Wong is affiliated with the Department of Electronic and Electrical Engineering, University College London, Torrington Place, WC1E 7JE, United Kingdom and he is also affiliated with the Department of Electronic Engineering, Kyung Hee University, Yongin-si, Gyeonggi-do 17104, Korea (e-mail: $\rm kai\text{-}kit.wong@ucl.ac.uk$).}
\blfootnote{\noindent
M. Kountouris is with the Department of Computer Science and Artificial Intelligence, Andalusian Research Institute in Data Science and Computational Intelligence (DaSCI), University
of Granada, Spain (e-mail: $\rm mariosk@ugr.es\rm$).}

\blfootnote{Corresponding author: Farshad Rostami Ghadi.}

\begin{abstract}
Robust learning in the presence of non-Gaussian and statistically dependent noise remains a fundamental challenge in signal processing and adaptive systems.
Although information-theoretic learning criteria such as correntropy offer strong robustness against impulsive and heavy-tailed disturbances, existing formulations are commonly applied componentwise and therefore do not explicitly exploit the dependence structures inherent in multivariate, multi-sensor, and temporal signals.
In this paper, we propose a learning framework, termed \textit{copula-induced information-theoretic learning} (CITL), which extends correntropy by embedding a copula space representation of residual dependence into the similarity measure. Unlike conventional correntropy-based approaches that operate pointwise on raw residuals, the proposed criterion is defined in a copula-transformed residual space, thus separating marginal robustness from dependence weighting. We derive a copula-induced correntropy (CIC) objective and a mixed marginal--dependence objective used in the implementation, provide information-theoretic and Bayesian interpretations,	and develop a robust conjugate gradient (CG) learning algorithm tailored to this criterion. For fixed smooth marginal estimators, a fixed copula-space metric, and a regularized radial penalty, we establish sufficient descent and global stationarity guarantees for the corresponding fixed-estimator subproblem under standard line-search conditions. Experiments on synthetic multivariate signal processing regression problems demonstrate that the proposed method consistently outperforms mean squared error (MSE), Huber, Student's-$t$, and classical correntropy-based approaches, 	particularly in the presence of dependent heavy-tailed noise.
\end{abstract}

\section{Introduction}

Robust learning under noise and uncertainty is a longstanding problem in signal processing (SP), adaptive filtering, and nonlinear system identification \cite{maronna2019robust}.
In classical supervised learning, model parameters are commonly estimated by minimizing the mean squared error (MSE) criterion \cite{hay2002adaptive,sun2017major,huber2011robust}.
Although MSE-based learning is optimal under the assumptions of independent and Gaussian noise,
these assumptions are rarely satisfied in practical SP applications~\cite{jav2025time}.
In many real-world scenarios, including multi-sensor fusion, vector time-series prediction, biomedical signal analysis, and financial modeling, noise sources (either due to measurement procedures or modeling errors) are often heavy-tailed, impulsive, and statistically dependent \cite{clavier2021imp}.
Under such conditions, MSE-based methods exhibit severe performance degradation due to their unbounded influence functions and excessive sensitivity to outliers.

To mitigate the impact of non-Gaussian disturbances, numerous robust alternatives to MSE have been proposed, including $L_p$-norm minimization \cite{nasri2009ad}, Huber-type losses \cite{zhu2018huber}, and Student's-$t$ likelihood-based formulations \cite{tang2024gen,huang2019novel}. While these methods improve robustness with respect to marginal outliers, they fundamentally rely on pointwise error modeling and largely neglect the dependence structure among error components \cite{gr2022an}. This limitation becomes particularly critical in multivariate and temporal SP problems,
where residuals exhibit strong cross-component, spatial, or temporal dependencies. Ignoring such dependence may lead to biased estimation, slow convergence, and poor generalization performance.

In recent years, information-theoretic learning (ITL) has emerged as a powerful paradigm for robust adaptive systems by leveraging similarity measures derived from information theory \cite{prin2010info}. Among ITL criteria, correntropy has attracted considerable attention due to its boundedness, locality, and inherent robustness against impulsive noise \cite{liu2007corr}.
Formally, correntropy measures the similarity between random variables through a kernel-induced expectation, thereby incorporating higher-order statistics beyond second-order moments. Despite these favorable properties, existing correntropy-based learning methods are predominantly formulated for scalar or componentwise residuals and implicitly assume independence among error components \cite{heravi2018new,chen2017rob}. Even more advanced variants such as generalized correntropy \cite{chen2016gen} improve robustness against heavy-tailed noise but still operate on marginal or componentwise residuals, without explicitly modeling statistical dependence among error components. As a consequence, classical correntropy fails to exploit the statistical dependence that naturally arises in multivariate, multi-sensor, and time-series signals.

One potential solution to this limitation is to employ copula theory to explicitly model statistical dependence structures. Generally speaking, copula theory offers a principled mathematical framework for modeling statistical dependence independently of marginal distributions \cite{nelson2006an,zeng2014copula}. By virtue of Sklar's theorem, any multivariate distribution can be decomposed into its marginal distributions and a copula function that fully characterizes the underlying dependence structure. This separation is particularly appealing for robust learning, as it enables heavy-tailed marginals and complex dependence patterns to be modeled in a decoupled manner.
Although copulas have been recently studied in wireless communications \cite{peter2014comm,ghadi2021copula,zheng2019copula,jor2021copula}, their integration into ITL and robust adaptive SP remains largely unexplored.

Motivated by these observations, we introduce a new learning paradigm that embeds copula-based dependence modeling directly into an information-theoretic similarity measure. Instead of applying correntropy pointwise to raw residuals, we define a novel criterion in a copula-transformed residual space, thereby capturing dependence-aware similarity while preserving robustness to heavy-tailed marginal disturbances. The resulting framework, termed
\emph{copula-induced information-theoretic learning} (CITL), provides a dependence extension of correntropy in the copula residual domain and yields a principled objective for robust learning under dependent noise.

The main contributions of this paper are summarized as follows:
\begin{itemize}
	\item We propose a copula-induced correntropy (CIC) criterion that  incorporates residual dependence through a copula-space metric while retaining robustness to heavy-tailed marginal noise.
	\item We introduce a mixed marginal-dependence objective that contains marginal correntropy and CIC as limiting cases and is used consistently in the algorithmic and numerical sections.
    \item We provide information-theoretic and Bayesian  interpretations of the proposed criterion; the latter is stated as a pseudo-likelihood/maximum a posteriori (MAP) analogy unless an explicitly normalized copula likelihood is specified.
	\item We develop a robust conjugate gradient (CG) learning algorithm tailored to the proposed objective, enabling efficient optimization for nonlinear and multivariate models.
	\item We provide a convergence analysis for the regularized fixed-estimator subproblem, including sufficient descent and global stationarity under strong Wolfe line search and an explicit bounded-direction safeguard.
	\item We demonstrate the effectiveness of the proposed framework through numerical experiments on synthetic multivariate signal processing regression problems.
\end{itemize}

The remainder of this paper is organized as follows. Section~\ref{sec:prelim} reviews background on correntropy, copula theory, and nonlinear CG methods. Section~\ref{sec:citl} introduces the proposed CIC criterion. Section~\ref{sec:opt} presents the corresponding CG learning algorithm. Section~\ref{sec:conv} provides convergence analysis. Section~\ref{sec:experiments} reports numerical results. Finally, Section~\ref{sec:conclusion} concludes the paper.

\subsection{Mathematical Notation and Conventions}
\label{subsec:notation}

Throughout this paper, scalars are denoted by lowercase letters $x$, vectors by bold lowercase letters $\mathbf x$, and matrices by bold uppercase letters $\mathbf X$. The $i$-th element of a vector $\mathbf x$ is denoted by $x_i$, and the $(i,j)$-th entry of a matrix $\mathbf X$ by $X_{ij}$. The superscript $(\cdot)^\top$ denotes transpose. For a vector $\mathbf x\in\mathbb{R}^p$ and a symmetric positive definite matrix $\mathbf A\in\mathbb{R}^{p\times p}$, the weighted norm is defined as
$
\|\mathbf x\|_{\mathbf A}^2 \triangleq \mathbf x^\top \mathbf A \mathbf x.
$
The Euclidean norm is denoted by $\|\cdot\|_2$.
The identity matrix of appropriate dimension is denoted by $\mathbf I$.

Random variables are denoted by uppercase letters $X$.
Expectation with respect to a random variable is denoted by $\mathbb E[\cdot]$. For a random vector $\mathbf E=(E_1,\ldots,E_p)$,
the marginal cumulative distribution function (CDF) and probability density function (PDF) of $E_i$ are denoted by $F_i(\cdot)$ and $f_i(\cdot)$, respectively. Given a set of input-output data samples $\{(\mathbf x_n,\mathbf d_n)\}_{n=1}^N$,
the output of a nonlinear model parameterized by $\mathbf w$ is denoted by
$\mathbf y_n = f(\mathbf x_n;\mathbf w)$,
and the corresponding residual vector is
$
\mathbf e_n = \mathbf d_n - \mathbf y_n
$. 
The copula-transformed residual vector is defined as
$
\mathbf u_n =
\big(F_1(e_{n,1}),\ldots,F_p(e_{n,p})\big)\in(0,1)^p
$. For a differentiable objective function $J(\mathbf w)$,
its gradient with respect to $\mathbf w$ is denoted by $\nabla J(\mathbf w)$. Unless otherwise stated, all gradients are assumed to be column vectors. The notation $\mathcal O(\cdot)$ denotes computational complexity up to constant factors. Finally, $\lambda_{\min}(\cdot)$ and $\lambda_{\max}(\cdot)$
denote the minimum and maximum eigenvalues of a symmetric matrix, respectively.

Table~\ref{tab:notation} summarizes the main symbols and parameters used throughout the paper.

\begin{table}[t]
	\caption{Summary of Notation}
	\label{tab:notation}
	\centering
	\begin{tabular}{ll}
		\hline
		\textbf{Symbol} & \textbf{Description} \\
		\hline
		$\mathbf{x}_n $ & Input vector at sample $n$ \\
		$\mathbf{d}_n$ & Desired (target) output vector \\
		$\mathbf{y}_n$ & Model output \\
		$\mathbf{w} $ & Model parameter vector \\
		$\mathbf{e}_n$ & Residual (error) vector \\
		$N$ & Number of training samples \\[2pt]

		$d$ & Input dimension \\
		$p$ & Output dimension \\
		$q$ & Number of model parameters \\[2pt]

		$F_i(\cdot)$ & Marginal CDF of the $i$-th residual component \\
		$f_i(\cdot)$ & Marginal PDF of the $i$-th residual component \\
		$\widehat F_i(\cdot)$ & Estimated marginal CDF \\
		$\widehat f_i(\cdot)$ & Estimated marginal PDF \\
		$\mathbf{u}_n$ & Copula-transformed residual vector \\
		$\mathbf{u}_0$ & Dependence center in copula space \\
		$\mathbf{s}_n$ & Centered copula residual \\
		$\rho_n$ & Squared Mahalanobis distance in copula space \\[2pt]

		$\Sigma$ & Copula-space dependence (covariance) matrix \\
		$\widehat{\Sigma}$ & Shrinkage estimate of $\Sigma$ \\
		$\lambda $ & Shrinkage parameter \\[2pt]

		$\alpha $ & Shape parameter controlling tail robustness \\
		  $\delta$ &   Smoothing constant in the copula space radial penalty\\
		$\sigma_{\rm k}$ & Marginal correntropy kernel size\\
		$\sigma_{\varepsilon}$ & Noise scale in the Student's-$t$ disturbance model\\
		$\omega_n^{(\delta)}$ & Regularized sample weight in CIC gradient \\[2pt]

		  $V_{\alpha,\Sigma,\delta}^{\mathrm{CIC}}$ & Copula-induced correntropy (CIC) \\
		  
          $J_\gamma(\mathbf{w})$ &   Mixed CIC objective function \\[2pt]

		  $\nabla J_\gamma(\mathbf{w})$ & Gradient of the  mixed CIC objective \\
		$\mathbf{J}_n$ & Jacobian of the model output \\[2pt]

		$\mathbf{w}_k$ & CG iterate at iteration $k$ \\
		$\mathbf{g}_k$ & Gradient at iteration $k$ \\
		  $\mathbf{p}_k$  & CG search direction \\
		$\alpha_k$ & Step size (line search) \\
		$\beta_k$ & CG update coefficient (PRP$^+$) \\
		$\eta$ & Restart threshold parameter \\[2pt]

		$c_1,c_2$ & Strong Wolfe line-search constants \\
		$R$ & Update period for marginal and dependence estimates \\
		$\epsilon$ & Copula space clipping constant \\[2pt]

		$\gamma$ & Mixing parameter between marginal and dependence terms\\
		\hline
	\end{tabular}
\end{table}

\section{Preliminaries and Background}
\label{sec:prelim}

This section reviews the theoretical foundations underpinning the proposed framework.
Specifically, we summarize correntropy-based ITL,
copula theory for dependence modeling, and nonlinear CG methods for optimization. The presentation emphasizes properties and limitations most relevant to the developments in subsequent sections.
\subsection{Correntropy and Information-Theoretic Learning}
\label{subsec:correntropy}

Let $X$ and $Y$ be random variables defined on a common probability space. Correntropy is an information-theoretic similarity measure defined as \cite{liu2007corr}
\begin{align}
	V_\sigma(X,Y)
	=
	\mathbb{E}\!\left[k_\sigma(X-Y)\right],
	\label{eq:correntropy_def}
\end{align}
where $k_\sigma(\cdot)$ is a shift-invariant positive definite kernel
with kernel size $\sigma>0$.


Unlike the MSE, which depends solely on second-order statistics, correntropy implicitly incorporates higher-order moments of the error distribution through the kernel-induced expectation. Moreover, due to the boundedness of $k_\sigma(\cdot)$, correntropy yields a similarity measure with a bounded influence function, thereby conferring inherent robustness against impulsive and heavy-tailed noise \cite{chen2016gen}.

Given $N$ paired observations $\{(x_n,y_n)\}_{n=1}^N$,
the empirical correntropy is estimated as \cite{liu2007corr}
\begin{equation}
	\widehat V_\sigma
	=
	\frac{1}{N}
	\sum_{n=1}^N
	k_\sigma(x_n-y_n).
	\label{eq:correntropy_empirical}
\end{equation}

Learning under the maximum correntropy criterion (MCC) \cite{chen2012max}
amounts to maximizing \eqref{eq:correntropy_empirical}
with respect to the model parameters. Despite its robustness advantages, classical correntropy is fundamentally a scalar or componentwise criterion. When extended to multivariate settings, it is typically applied independently to each error component, thereby implicitly assuming statistical independence among residuals. As a consequence, existing correntropy-based learning methods fail to exploit dependence structures that are ubiquitous in multivariate, multi-sensor, and temporal SP problems.

\subsection{Copula Theory for Dependence Modeling}
\label{subsec:copula}

Let $\mathbf E=(E_1,\ldots,E_p)$ be a $p$-dimensional random vector
with marginal CDFs $F_i(e_i)$.
According to Sklar's theorem,
the joint CDF of $\mathbf E$ can be expressed as \cite{nelson2006an}
\begin{equation}
	F_{\mathbf E}(e_1,\ldots,e_p)
	=
	C\!\left(F_1(e_1),\ldots,F_p(e_p)\right),
	\label{eq:sklar}
\end{equation}
where $C:[0,1]^p\rightarrow[0,1]$ is a copula function.

The copula $C(\cdot)$ uniquely characterizes the dependence structure among the components of $\mathbf E$, independently of the marginal distributions. This separation between marginals and dependence is particularly advantageous for robust signal processing, as it enables heavy-tailed or impulsive marginal behavior to be modeled independently of cross-component or temporal dependence. Full copula models can capture nonlinear and tail dependence. The covariance-based implementation adopted in this paper should instead be interpreted as a regularized summary of dependence directions in copula space, rather than as a full copula density model. Despite these advantages, copula theory has seen limited integration
with ITL criteria, particularly in the context of robust adaptive algorithms.

\subsection{Nonlinear Conjugate Gradient Methods}
\label{subsec:cg_background}

Consider the unconstrained optimization problem
\begin{equation}
	\min_{\mathbf w \in \mathbb R^q} J(\mathbf w),
	\label{eq:unconstrained_problem}
\end{equation}
where $J:\mathbb R^q\rightarrow\mathbb R$ is continuously differentiable.

Nonlinear CG methods generate a sequence of iterates
$\{\mathbf w_k\}$ according to \cite{no2006num}
\begin{equation}
	\mathbf w_{k+1} = \mathbf w_k + \alpha_k \mathbf p_k,
	\label{eq:cg_update_bg}
\end{equation}
where $\alpha_k>0$ is a step size obtained via line search,
and $\mathbf p_k$ is the search direction defined recursively as \cite{no2006num}
\begin{equation}
	\mathbf p_k = -\nabla J(\mathbf w_k) + \beta_k \mathbf p_{k-1},
	\qquad
	\mathbf p_0 = -\nabla J(\mathbf w_0).
	\label{eq:cg_direction_bg}
\end{equation}

Different choices of the scalar parameter $\beta_k$ give rise to well-known CG variants. For example, the Fletcher-Reeves update is defined as \cite{chat2010fle}
\begin{align}
\beta_k^\mathrm{FR} = \frac{\|\nabla J(\mathbf {w}_k)\|_2^2}{\|\nabla J(\mathbf {w}_{k-1})\|_2^2},
\end{align}
the Polak-Ribi\`ere-Polyak (PRP) \cite{zhang2006des} update is given by
\begin{align}
\beta_k^\mathrm{PRP} = \frac{\nabla J(\mathbf {w}_k)^\top(\nabla J(\mathbf {w}_k)-\nabla J(\mathbf {w}_{k-1}))}{\|\nabla J(\mathbf {w}_{k-1})\|_2^2},
\end{align}
and the Hager-Zhang update is described in \cite{noc2005new}.
 When combined with appropriate line search strategies, such as the strong Wolfe conditions, these methods guarantee sufficient descent and global convergence for a broad class of nonconvex objective functions. It should be noted that due to their favorable convergence properties and low memory requirements, nonlinear CG methods have been widely adopted in large-scale signal processing and adaptive learning problems. In this work, these methods provide an efficient optimization framework for the proposed CIC objective.

\section{Copula-Induced Information-Theoretic Learning}
\label{sec:citl}

This section introduces the proposed CITL framework. We begin by formalizing the problem of dependence-aware error modeling in SP, then define a novel CIC criterion, and finally establish its fundamental properties and interpretations.

\subsection{Error Dependence Modeling in Signal Processing}
\label{subsec:error_modeling}

Consider a general nonlinear adaptive system or regression model of the form
\begin{align}
	\mathbf y_n = f(\mathbf x_n;\mathbf w), \qquad \mathbf y_n \in \mathbb{R}^p,
\end{align}
where $\mathbf x_n \in \mathbb{R}^d$ denotes the input signal, $\mathbf w \in \mathbb{R}^q$ is the parameter vector, and $\mathbf d_n \in \mathbb{R}^p$ denotes the desired response. Therefore, the corresponding residual vector is defined as
\begin{align}
	\mathbf e_n = \mathbf d_n - \mathbf y_n.
\end{align}

In a wide range of SP applications, including multi-sensor systems, vector time-series modeling, and multi-output nonlinear systems, the components of $\mathbf e_n$ are neither independent nor identically distributed. Correlations induced by shared interference sources, temporal coupling, or physical constraints lead to structured dependence that cannot be adequately captured by pointwise error modeling. Conventional learning criteria such as MSE or classical correntropy implicitly rely on independence assumptions and therefore fail to exploit the information contained in such dependence structures. This observation motivates the development of a learning criterion that explicitly incorporates statistical dependence among residual components.

\subsection{Copula-Transformed Residual Space}
\label{subsec:copula_space}

Let $F_i(\cdot)$ denote the marginal CDF of the $i$-th component of the residual vector. We define the copula-transformed residual vector as
\begin{align}
	\mathbf u_n =
	\big(
	F_1(e_{n,1}), F_2(e_{n,2}), \ldots, F_p(e_{n,p})
	\big) \in (0,1)^p.\label{eq:cop-trans}
\end{align}
The transformation in \eqref{eq:cop-trans} maps each residual component to a uniformly distributed variable on the unit interval, thereby removing all marginal distributional information. Consequently, the joint distribution of $\mathbf u_n$ encodes only the dependence structure among the residual components.

By Sklar's theorem, the joint distribution of $\mathbf e_n$ can be expressed in terms of the copula function $C(\cdot)$ as
\begin{align}
	F_{\mathbf{E}}(\mathbf e_n)
	=
	C(\mathbf u_n),
\end{align}
where $C:[0,1]^p \rightarrow [0,1]$ uniquely characterizes the dependence structure of the marginals. Operating in the copula-transformed space thus enables robust  treatment of heavy-tailed marginal behavior while preserving dependence information. In the present implementation, this dependence information is summarized through a regularized covariance metric in copula space; richer copula density models would be required to fully characterize nonlinear or tail dependence.

\subsection{Copula-Induced Correntropy}
\label{subsec:cic}

We now introduce a new information-theoretic similarity measure defined directly on the copula-transformed residuals. In what follows, we specialize the kernel function to the Gaussian kernel, defined as \cite{prin2010info}
\begin{equation}
k_\sigma(e) = \exp\!\left(-\frac{e^2}{2\sigma^2}\right),
\label{eq:gaussian_kernel}
\end{equation}
due to its smoothness, boundedness, and analytical tractability. This choice leads to an exponential-type similarity measure in the copula-transformed residual space.

\begin{definition}[Copula-Induced Correntropy]
Let $\mathbf u_n \in (0,1)^p$ denote the copula-transformed residual vector   and define \begin{equation}
	\rho_n=(\mathbf u_n-\mathbf u_0)^\top\Sigma^{-1}(\mathbf u_n-\mathbf u_0),
\end{equation} 
where $\mathbf u_0 \in (0,1)^p$ denotes the dependence center, typically chosen as $\mathbf u_0=\tfrac12\mathbf 1$,  and $\Sigma \succ 0$ is a positive definite dependence matrix  in copula space.
The regularized CIC is defined as
  \begin{equation}
	V_{\alpha,\Sigma,\delta}^{\mathrm{CIC}}
	=
	\mathbb{E}
	\left[
	\exp\!\left(
	-
	(\rho_n+\delta)^{\alpha/2}
	\right)
	\right],
	\label{eq:CIC_def}
\end{equation}
where $\alpha \in (0,2]$ controls tail robustness and $\delta\ge0$ is a smoothing constant.
The unregularized form used for the basic definition is recovered by setting $\delta=0$, whereas the convergence analysis for $\alpha<2$ uses $\delta>0$ to avoid a singular derivative at $\rho_n=0$.
\end{definition}

Unlike classical correntropy, which measures similarity directly in the raw error domain,
the proposed CIC measures similarity in the dependence space.
As a result, structured deviations arising from coherent or dependent disturbances
are treated fundamentally differently from isolated outliers.
Moreover, the parameter $\alpha$ provides explicit control over robustness, namely, smaller values of $\alpha$   flatten the radial penalty for large deviations and can improve resilience to heavy-tailed noise. For $\alpha<2$, however, the derivative of the unregularized penalty is singular at $\rho_n=0$; hence, the regularized form with $\delta>0$ is used in the theoretical analysis. In practice, the copula-transformed residuals are obtained by first computing the residuals and then applying estimated marginal CDFs. No parametric copula family is assumed; instead, dependence is captured through covariance estimation in copula space.

\subsection{Empirical Estimation and Learning Objective}
\label{subsec:objective}

Given a finite set of $N$ samples $\{(\mathbf x_n,\mathbf d_n)\}_{n=1}^N$, the empirical  CIC estimate is \begin{align}
	\widehat V_{\alpha,\Sigma,\delta}^{\mathrm{CIC}}(\mathbf w)
	=
	\frac{1}{N}
	\sum_{n=1}^N
	\exp\!\left[-(\rho_n(\mathbf w)+\delta)^{\alpha/2}\right],
	\label{eq:CIC_empirical}
\end{align} 
where $\rho_n(\mathbf w)=(\mathbf u_n(\mathbf w)-\mathbf u_0)^\top\Sigma^{-1}(\mathbf u_n(\mathbf w)-\mathbf u_0)$.
To match the implementation used in the numerical section and to retain explicit marginal robustness, we optimize the mixed marginal--dependence objective \begin{align}\nonumber
	&J_\gamma(\mathbf w)
	=\\
    &	-\frac{1}{N}\sum_{n=1}^{N}
	\exp\!\left[-(1-\gamma)\psi_{\rm marg}(\mathbf e_n(\mathbf w))-\gamma\psi_{\rm dep}(\mathbf u_n(\mathbf w))\right],
	\label{eq:CIC_objective}
\end{align} 
  with\begin{align}
	\psi_{\rm marg}(\mathbf e_n)=\frac{\|\mathbf e_n\|_2^2}{2\sigma_{\rm k}^2},
	\qquad
	\psi_{\rm dep}(\mathbf u_n)=(\rho_n+\delta)^{\alpha/2},
	\label{eq:marg_dep_penalties}
\end{align} 
where $\gamma\in[0,1]$ balances marginal correntropy and copula space dependence modeling.
The limiting case $\gamma=1$ gives the pure CIC objective, whereas $\gamma=0$ gives a multivariate Gaussian-kernel correntropy objective in the raw residual domain.
For notational compactness, the remainder of the paper writes $J(\mathbf w)\equiv J_\gamma(\mathbf w)$ unless a limiting case is explicitly discussed. The resulting objective   is bounded and generally nonconvex. Its robustness arises from the combination of marginal kernel weighting, the bounded copula transformation, and the dependence weighting induced by the radial penalty defined through $\Sigma$.

\subsection{Fundamental Properties}
\label{subsec:properties}

\begin{Proposition}
For all $\mathbf w$, the empirical CIC satisfies
 \begin{align}
	0 < \widehat V_{\alpha,\Sigma,\delta}^{\mathrm{CIC}}(\mathbf w) \le 1.
\end{align} 
\end{Proposition}

\begin{proof}
The exponential kernel in \eqref{eq:CIC_def} is strictly positive and upper bounded by one, which directly implies the stated result.
\end{proof}

\begin{remark}
The proposed criterion defines a correntropy criterion in the copula-transformed residual domain, not as an exact algebraic reduction of classical correntropy in the raw error domain.
In the scalar case $p=1$, CIC applies the kernel to $F(e)-u_0$ rather than to $e$.
Thus, it locally resembles generalized correntropy in the raw residual only when the residual distribution is symmetric around zero, $F(0)=u_0=1/2$, and the marginal CDF is locally linear, $F(e)-1/2\approx f(0)e$, in the small error region.
Similarly, when $\alpha=2$ and $\Sigma=\mathbf I$, CIC  becomes a Gaussian kernel in copula space; it coincides with Gaussian MCC in the raw residual domain only under the same local linear CDF approximation and an appropriate effective kernel width.
\end{remark}



\subsection{Bayesian Interpretation}
\label{subsec:bayesian}

\begin{remark}
The proposed objective admits a Bayesian  interpretation rather than, in its present form, a fully normalized Bayesian likelihood. The exponential term in \eqref{eq:CIC_objective} can be viewed as an unnormalized pseudo-likelihood that assigns low influence to samples with large marginal errors or large copula space deviations. A strict MAP derivation would require specifying normalized marginal densities and a normalized copula density, including the associated Jacobian terms. Accordingly, the proposed objective is best interpreted as a smoothed MAP analogue with bounded kernel influence, in the same spirit as maximum correntropy estimation \cite{chen2012max}, but adapted to a copula-transformed residual representation.
\end{remark}

\section{Conjugate-Gradient Optimization of the Copula-Induced Correntropy Objective}
\label{sec:opt}

This section develops a numerically stable optimization framework for minimizing the mixed CIC objective \eqref{eq:CIC_objective}. The proposed learning algorithm is based on nonlinear CG methods, which provide fast convergence, low memory complexity, and strong theoretical guarantees when combined with standard line-search conditions. A central technical challenge arises from the copula transformation, which is defined through marginal CDFs. To enable gradient-based optimization and backpropagation, the mapping $\mathbf w \mapsto \mathbf u_n(\mathbf w)$ must be differentiable. We therefore adopt smooth marginal CDF estimators or parametric marginals whose derivatives are valid PDFs. Under this setting, and with the marginal estimators and copula-space metric held fixed within each line-search block, the resulting objective admits a closed-form gradient, which enables the construction of a principled CG learning rule with convergence guarantees for the fixed-estimator subproblem.

\subsection{Differentiable Copula Transform via Smooth Marginal CDF Estimation}
\label{subsec:smoothcdf}

Let $\mathbf e_n(\mathbf w)=\mathbf d_n-f(\mathbf x_n;\mathbf w)\in\mathbb{R}^p$
denote the residual vector.
For each component $e_{n,i}$, we first define a smooth estimate of its marginal
probability density function (PDF) as
\begin{align}
	\widehat f_i(e)
	=
	\frac{1}{Nh_i}
	\sum_{m=1}^N
	K\!\left(\frac{e-e_{m,i}}{h_i}\right),
\end{align}
and the corresponding CDF is then formulated as
\begin{align}
	\widehat F_i(e)
	=
	\int_{-\infty}^{e} \widehat f_i(t)\,dt, \label{eq:kde_cdf_pdf}
\end{align}
where $K(\cdot)$ is a smooth kernel function, e.g., Gaussian, and $h_i>0$ is a bandwidth parameter. Kernel density estimation ensures that $\widehat F_i(\cdot)$ is continuously differentiable \cite{silv2018dens}. In the gradient formulas below, $\widehat F_i$ and $\widehat f_i$ are treated as fixed during each CG/line-search block; if they are differentiated through as functions of all residuals, additional cross-sample derivative terms arise. The copula-transformed residual is then given by
\begin{equation}
	\mathbf u_n
	=
	\big(\widehat F_1(e_{n,1}),\ldots,\widehat F_p(e_{n,p})\big)^\top
	\in(0,1)^p,
	\label{eq:u_def}
\end{equation}
so that its Jacobian with respect to $\mathbf e_n$ is diagonal, i.e., 
\begin{equation}
	\frac{\partial \mathbf u_n}{\partial \mathbf e_n}
	=
	\mathbf D_n
	=
	\mathrm{diag}
	\big(\widehat f_1(e_{n,1}),\ldots,\widehat f_p(e_{n,p})\big).
	\label{eq:Dn_def}
\end{equation}

For numerical stability, we apply a standard copula-space clipping
\begin{equation}
	u_{n,i} \leftarrow \min\{1-\epsilon,\max\{\epsilon,u_{n,i}\}\},
	\label{eq:clip}
\end{equation}
with a small $\epsilon>0$, e.g., $10^{-6}$, which prevents boundary effects and avoids ill-conditioned dependence estimates. Hard clipping is not differentiable at the clipping thresholds; therefore, the convergence analysis assumes that clipping is inactive on the level set or is replaced by a smooth clipping approximation. In implementation, hard clipping may be used with the clipped values treated as fixed for the purpose of each line search. In batch learning, $\{\widehat F_i,\widehat f_i\}$ can be updated every $R$ iterations using the current residuals. In streaming settings, $\widehat f_i$ can be updated recursively using exponential forgetting. If parametric marginals are preferred, e.g., Student's-$t$, one may set $\widehat F_i$ to the parametric CDF and $\widehat f_i$ to its PDF, yielding identical derivations.

\subsection{Dependence Metric in Copula Space}
\label{subsec:sigma_choice}

The matrix $\Sigma \succ 0$ defines the metric in copula space and determines how the criterion penalizes deviations along different dependence directions. In this work, $\Sigma$ is estimated from copula residuals via a shrinkage covariance estimator \cite{leo2004well,chen2010sh},
\begin{align}
	\widehat \Sigma
	=
	(1-\lambda)\,\widehat{\mathrm{Cov}}(\mathbf u_n)
	+
	\lambda\,\mathrm{tr}\!\left(\widehat{\mathrm{Cov}}(\mathbf u_n)\right)\frac{\mathbf I}{p}
	+ \varepsilon_\Sigma \mathbf I,
	\label{eq:shrinkage_sigma}
\end{align}
where $\lambda \in [0,1]$ is a shrinkage parameter and $\varepsilon_\Sigma>0$ is a small ridge constant. This ridge term guarantees strict positive definiteness even when the empirical covariance is rank deficient; without it, shrinkage alone guarantees positive definiteness only under additional conditions such as $\lambda>0$ and nonzero trace. Since the covariance is estimated from the 
copula-transformed residuals $\mathbf{u}_n \in (0,1)^p$, the marginal heavy-tailed effects are mitigated by the copula transformation, while the shrinkage estimator provides a regularized estimate of the dependence structure. In practice, $\widehat \Sigma$ may be updated every $R$ iterations together with the marginal estimators. For applications with known dependence structure, e.g., banded temporal correlations, $\Sigma$ may be constrained accordingly.


\subsection{Structure of the Objective Function}
\label{subsec:objective_structure}

For fixed marginal estimators and a fixed copula-space metric during a CG block, define
\begin{align}
&\mathbf s_n = \mathbf u_n-\mathbf u_0, \quad \rho_n = \mathbf s_n^\top\Sigma^{-1}\mathbf s_n,\label{eq:sn_rhon}\\
&\psi_{\rm marg,n} = \frac{\|\mathbf e_n\|_2^2}{2\sigma_{\rm k}^2}, \quad \psi_{\rm dep,n} = (\rho_n+\delta)^{\alpha/2},\label{eq:psi_defs}\\
	&\kappa_n = \exp\!\left[-(1-\gamma)\psi_{\rm marg,n}-\gamma\psi_{\rm dep,n}\right].\label{eq:kappa_n}
\end{align} 
Then, the optimized empirical objective is \begin{align}
	J(\mathbf w)=-\frac{1}{N}\sum_{n=1}^N \kappa_n.
	\label{eq:J_recall}
\end{align} 
The smoothing constant $\delta>0$ is only needed when $\alpha<2$ is used in the convergence analysis; it may be chosen very small in practice.

\subsection{Gradient Derivation}
\label{subsec:gradient}
Differentiating \eqref{eq:J_recall} with respect to the model parameters yields two contributions: a marginal correntropy term and a copula-space dependence term.
With $\mathbf J_n=\partial\mathbf y_n/\partial\mathbf w$ and $\partial\mathbf e_n/\partial\mathbf w=-\mathbf J_n$, the gradient is 
\begin{align}\notag
	&\nabla J(\mathbf w)
	=
	-\frac{1}{N}\\
    &\times\sum_{n=1}^{N}\kappa_n
	\left[
	\frac{1-\gamma}{\sigma_{\rm k}^2}\mathbf J_n^\top\mathbf e_n
	+
	\gamma\alpha(\rho_n+\delta)^{\frac{\alpha}{2}-1}
	\mathbf J_n^\top\mathbf D_n\Sigma^{-1}\mathbf s_n
	\right].
	\label{eq:grad_final}
\end{align} 
Equivalently, the copula space contribution can be written using the regularized weight \begin{align}
	\omega_n^{(\delta)}
	=
	\frac{\alpha}{2}(\rho_n+\delta)^{\frac{\alpha}{2}-1}\kappa_n,
	\label{eq:omega_def}
\end{align} 
which remains finite at $\rho_n=0$ whenever $\delta>0$.
When $\gamma=1$ and $\delta=0$, \eqref{eq:grad_final} reduces to the unregularized pure-CIC gradient; however, for $\alpha<2$ the associated scalar factor is singular at $\rho_n=0$, which is why the regularized form is used for the theoretical guarantees. Equation \eqref{eq:grad_final} explicitly shows how marginal residual magnitude, marginal density scaling through $\mathbf D_n$, copula space dependence through $\Sigma^{-1}\mathbf s_n$, and model sensitivity through $\mathbf J_n$  jointly determine the  update.

\subsection{Nonlinear Conjugate Gradient Learning}
\label{subsec:cg}

Here, we minimize $J(\mathbf w)$ using a nonlinear CG iteration
\begin{equation}
	\mathbf w_{k+1}
	=
	\mathbf w_k
	+
	\alpha_k \mathbf p_k,
	\qquad
	\mathbf p_k
	=
	-\mathbf g_k
	+
	\beta_k \mathbf p_{k-1},
	\label{eq:cg_update}
\end{equation}
where $\mathbf g_k=\nabla J(\mathbf w_k)$ and $\mathbf p_0=-\mathbf g_0$.
We adopt the PRP update with nonnegativity safeguard, i.e., PRP$^+$ \cite{gilbert1992glob}, as follows
\begin{equation}
	\beta_k^{\mathrm{PRP}^+}
	=
	\max\!\left\{
	0,\;
	\frac{\mathbf g_k^\top(\mathbf g_k-\mathbf g_{k-1})}{\|\mathbf g_{k-1}\|_2^2}
	\right\}.
	\label{eq:beta_prpplus}
\end{equation}

To preserve descent under nonconvexity and approximate gradients, we apply a restart whenever the sufficient descent condition is violated, i.e.,
\begin{align}
	\text{if }\ \mathbf g_k^\top \mathbf p_k \ge -\eta\|\mathbf g_k\|_2^2,\ \text{ then set }\ \mathbf p_k=-\mathbf g_k,
	\label{eq:restart}
	\end{align}
where  $\eta\in(0,1)$.
\subsection{Line Search and Step Size Selection}
\label{subsec:linesearch}
Step sizes $\alpha_k$ are selected using a strong Wolfe line search, ensuring sufficient decrease and curvature conditions \cite{no2006num}, i.e.,
\begin{align}
	J(\mathbf w_k+\alpha_k\mathbf p_k)
	&\le
	J(\mathbf w_k)
	+
	c_1\alpha_k\mathbf g_k^\top\mathbf p_k,
	\label{eq:wolfe1}
	\\
	\big|
	\nabla J(\mathbf w_k+\alpha_k\mathbf p_k)^\top\mathbf p_k
	\big|
	&\le
	c_2\big|\mathbf g_k^\top\mathbf p_k\big|,
	\label{eq:wolfe2}
\end{align}
where $0<c_1<c_2<1$.
These conditions are standard in nonlinear CG analysis \cite{no2006num}
and will be invoked directly in Section~\ref{sec:conv}.

\subsection{Algorithm Description}
\label{subsec:alg}

Algorithm~\ref{alg:citl_cg} summarizes the complete CIC-CG learning procedure for the objective \eqref{eq:CIC_objective}.
The method alternates between (i) transforming residuals to the copula space,
(ii) updating the copula-space dependence metric via shrinkage covariance estimation,
and (iii) performing a safeguarded nonlinear CG step with a strong Wolfe line search.
The periodic refresh of the marginal CDF/PDF estimators and $\widehat{\Sigma}$ improves numerical stability under heavy-tailed disturbances. The descent and convergence guarantees stated in Section~\ref{sec:conv} apply to each fixed estimator CG block; when the estimators are refreshed, the optimized objective is correspondingly updated.

\begin{algorithm}[t]
	\caption{CIC-CG: Copula-Induced Correntropy Conjugate Gradient Learning}
	\label{alg:citl_cg}
	\begin{algorithmic}[1]
		\STATE \textbf{Input:} $\{(\mathbf x_n,\mathbf d_n)\}_{n=1}^N$, $\alpha$, $\gamma$, $\sigma_{\rm k}$, $\delta$, $\lambda$, $\varepsilon_\Sigma$, $\mathbf u_0$, initial $\mathbf w_0$, update period $R$
		\STATE \textbf{Initialize:} estimate $\widehat F_i,\widehat f_i$ using \eqref{eq:kde_cdf_pdf}; compute $\mathbf u_n$ using \eqref{eq:u_def}; apply clipping or smooth clipping as in \eqref{eq:clip}
		\STATE Estimate $\Sigma$ using \eqref{eq:shrinkage_sigma}; freeze $\widehat F_i$, $\widehat f_i$, and $\Sigma$ for the current CG block
		\STATE Compute $\mathbf g_0=\nabla J(\mathbf w_0)$ via \eqref{eq:grad_final}; set $\mathbf p_0=-\mathbf g_0$
		\FOR{$k=0,1,2,\ldots$ until stopping criterion}
		\IF{$k>0$ and $\mathrm{mod}(k,R)=0$} \STATE Update $\widehat F_i,\widehat f_i$ and recompute $\mathbf u_n$; update $\Sigma$ via \eqref{eq:shrinkage_sigma}; recompute $\mathbf g_k=\nabla J(\mathbf w_k)$ for the refreshed fixed objective; restart the block with $\mathbf p_k=-\mathbf g_k$. \ENDIF
		\STATE Compute $\alpha_k$ via strong Wolfe line search \eqref{eq:wolfe1}-\eqref{eq:wolfe2} for the current fixed objective \STATE Update $\mathbf w_{k+1}=\mathbf w_k+\alpha_k\mathbf p_k$
		\STATE Compute $\mathbf g_{k+1}=\nabla J(\mathbf w_{k+1})$ using \eqref{eq:grad_final}
		\STATE Compute $\beta_{k+1}$ by \eqref{eq:beta_prpplus} and set $\mathbf p_{k+1}=-\mathbf g_{k+1}+\beta_{k+1}\mathbf p_k$
		\IF{$\mathbf g_{k+1}^\top \mathbf p_{k+1} \ge -\eta\|\mathbf g_{k+1}\|_2^2$ or $\|\mathbf p_{k+1}\|_2>M_p\|\mathbf g_{k+1}\|_2$}
		\STATE Restart/safeguard: $\mathbf p_{k+1}=-\mathbf g_{k+1}$
		\ENDIF
		\ENDFOR
		\STATE \textbf{Output:} $\mathbf w_\star$
	\end{algorithmic}
\end{algorithm}

\subsection{Computational Complexity}
\label{subsec:complexity}

Now, let $q$ be the number of parameters and $p$ the output dimension. Once marginal CDFs are available, computing $\{\mathbf u_n\}$ requires $O(Np)$ operations. Gradient evaluation \eqref{eq:grad_final} is dominated by Jacobian-vector products $\mathbf J_n^\top(\cdot)$, which are obtained by backpropagation at cost comparable to MSE learning. The additional copula-space overhead is primarily the multiplication $\Sigma^{-1}\mathbf s_n$. If $\Sigma^{-1}$ is dense, this costs $O(p^2)$ per sample; if $\Sigma$ is diagonal, banded, or low-rank, this cost is reduced accordingly. Updating the KDE marginals and shrinkage covariance introduces an additional periodic cost controlled by $R$, which is not incurred at every CG step when the estimators are frozen within each block.

\section{Convergence Analysis}
\label{sec:conv}
This section establishes convergence properties for the regularized mixed objective \eqref{eq:CIC_objective} optimized by the CIC-CG method.
Because the empirical marginal CDF/PDF estimates and the shrinkage covariance may be refreshed during training, the objective is generally nonstationary across refreshes.
Accordingly, the theorem below applies to a fixed estimator CG block, i.e., with $\widehat F_i$, $\widehat f_i$, and $\Sigma$ held fixed during the line search and CG updates.
If these quantities are updated periodically, the result applies within each block; convergence of the full nonstationary sequence would require additional assumptions on the magnitude and frequency of the estimator updates.

\subsection{Assumptions and Preliminaries}
\label{subsec:assumptions}

\begin{assumption}[Bounded Level Set]\label{assum:1}
Let
\begin{align}
\mathcal L = \{\mathbf w \in \mathbb R^q : J(\mathbf w) \le J(\mathbf w_0)\}
\end{align}
denote the level set associated with the initial iterate $\mathbf w_0$ for the fixed estimator objective.
It is assumed that $\mathcal L$ is bounded.
 \end{assumption}

 \begin{assumption}[Smooth Model Mapping]\label{assum:2}
The nonlinear mapping $f(\mathbf x;\mathbf w)$ is continuously differentiable with respect to $\mathbf w$ on $\mathcal L$.
Moreover, its Jacobian $\mathbf J_n(\mathbf w)$ is bounded and  Lipschitz continuous on $\mathcal L$; that is, there exists a constant $L_f>0$ such that
  \begin{align}
\|\mathbf J_n(\mathbf w_1)-\mathbf J_n(\mathbf w_2)\|
\le
L_f\|\mathbf w_1-\mathbf w_2\|,
\quad
\forall\,\mathbf w_1,\mathbf w_2\in\mathcal L.
\end{align} 
 \end{assumption}

\begin{assumption}[Fixed Smooth Marginal Transform]\label{assum:3}
  Within the CG block under analysis, the marginal estimators $\widehat F_i(\cdot)$ and $\widehat f_i(\cdot)$ are fixed.
The densities $\widehat f_i(\cdot)$  are bounded and Lipschitz continuous on $\mathbb R$, i.e., there exists  $M_f>0$ such that \begin{align}
0 \le \widehat f_i(e) \le M_f,
\qquad
\forall\, e\in\mathbb R,
\quad \forall i.
\end{align} 
 If clipping is used, it is either inactive on $\mathcal L$ or replaced by a smooth clipping map.
\end{assumption}

\begin{assumption}[Positive-Definite Dependence Metric]\label{assum:4}
  Within the CG block under analysis, $\Sigma$ is fixed, symmetric positive definite,   and satisfies \begin{align}
0 < \lambda_{\min}(\Sigma) \le \lambda_{\max}(\Sigma) < \infty.
\end{align}
  \end{assumption}

\begin{assumption}[Regularized Radial Penalty]\label{assum:5}
Either $\alpha=2$ with $\delta\ge0$, or $0<\alpha<2$ with $\delta>0$.
This condition ensures that $(\rho_n+\delta)^{\alpha/2}$ has a bounded derivative with respect to $\rho_n$ on $\rho_n\ge0$.
\end{assumption}

\begin{assumption}[Bounded Search Directions]\label{assum:6} The generated search directions satisfy $\|\mathbf p_k\|_2\le M_p\|\mathbf g_k\|_2$ for some constant $M_p>0$ whenever $\mathbf g_k\ne0$. This bounded angle condition is imposed in Zoutendijk convergence results. In Algorithm~\ref{alg:citl_cg}, it is enforced explicitly by restarting whenever $\|\mathbf p_k\|_2>M_p\|\mathbf g_k\|_2$; the descent restart alone would not guarantee this uniform bound.
\end{assumption}

\subsection{Lipschitz Continuity of the Gradient}
\label{subsec:lipschitz}

\begin{lemma}[Bounded Regularized Kernel Weight]\label{lemma:1}
  Under Assumption~\ref{assum:5}, the scalar factor
  \begin{align}
\omega_n^{(\delta)}
=
\frac{\alpha}{2}(\rho_n+\delta)^{\frac{\alpha}{2}-1}\kappa_n
\end{align} 
is bounded for all $\rho_n\ge0$. \end{lemma}
\begin{proof}
 Since $0<\kappa_n\le1$, it is sufficient to bound $(\rho_n+\delta)^{\alpha/2-1}$.
If $\alpha=2$, this factor equals one.
If $0<\alpha<2$ and $\delta>0$, then
\begin{align}
(\rho_n+\delta)^{\alpha/2-1}
\le
\delta^{\alpha/2-1},
\end{align}
because the exponent $\alpha/2-1$ is negative. Thus $\omega_n^{(\delta)}$ admits a finite upper bound. Without the regularization, i.e., $\delta=0$, the factor is unbounded as $\rho_n\to0^+$ for every $\alpha<2$. \end{proof}

\begin{lemma}[Lipschitz Continuity of the Gradient]\label{lemma:2}
Under Assumptions \ref{assum:1}-  \ref{assum:5}, the gradient $\nabla J(\mathbf w)$ of the fixed-estimator objective is Lipschitz continuous on $\mathcal L$.
That is, there exists  $L_J>0$ such that
 \begin{align}
\|\nabla J(\mathbf w_1)-\nabla J(\mathbf w_2)\|
\le
L_J\|\mathbf w_1-\mathbf w_2\|,
\quad
\forall\,\mathbf w_1,\mathbf w_2\in\mathcal L.
\end{align} 
\end{lemma}
 \begin{proof}
From \eqref{eq:grad_final},  $\nabla J(\mathbf w)$ is  a finite sum of products involving $\kappa_n$, $\mathbf J_n(\mathbf w)$, $\mathbf e_n(\mathbf w)$, $\mathbf D_n(\mathbf w)$,   $\Sigma^{-1}\mathbf s_n(\mathbf w)$, and the regularized radial factor $(\rho_n+\delta)^{\alpha/2-1}$.
On the bounded level set, Assumptions \ref{assum:2}-\ref{assum:4} imply boundedness and Lipschitz continuity of the model, marginal transform, and metric terms.
Lemma~\ref{lemma:1} and Assumption~\ref{assum:5} remove the singularity that would otherwise occur at $\rho_n=0$ for $\alpha<2$. Finite sums and products of bounded Lipschitz functions are Lipschitz continuous; hence the result follows.
\end{proof}

\subsection{Sufficient Descent Property}
\label{subsec:descent}

\begin{lemma}[Sufficient Descent]\label{lemma:3}
Let $\mathbf p_k$ be the search direction generated by the PRP$^+$ update with the restart rule in \eqref{eq:restart}.
Then
 \begin{align}
\mathbf g_k^\top \mathbf p_k
\le
-\eta \|\mathbf g_k\|_2^2,
\qquad
\forall k.
\end{align} 
\end{lemma}
\begin{proof}
  If the candidate PRP$^+$   direction violates the descent test in \eqref{eq:restart}, the algorithm sets $\mathbf p_k=-\mathbf g_k$, giving $\mathbf g_k^\top\mathbf p_k=-\|\mathbf g_k\|_2^2\le-\eta\|\mathbf g_k\|_2^2$ because $0<\eta<1$. If the restart is not triggered, the negation of the test gives $\mathbf g_k^\top\mathbf p_k< -\eta\|\mathbf g_k\|_2^2$ directly.
\end{proof}

 \subsection{Global Convergence}
\label{subsec:global}

\begin{theorem}[Global Stationarity for a Fixed Estimator Block]
Suppose that Assumptions \ref{assum:1}-\ref{assum:6} hold.
Let $\{\mathbf w_k\}$ be the sequence generated by Algorithm~\ref{alg:citl_cg} within a fixed estimator block, with step sizes satisfying the strong Wolfe conditions \eqref{eq:wolfe1}-\eqref{eq:wolfe2}.
Then, for that fixed objective,
\begin{align}
\lim_{k\rightarrow\infty} \|\nabla J(\mathbf w_k)\|_2 = 0
\end{align}
provided the fixed-estimator subproblem is optimized for infinitely many CG iterations. When the practical algorithm uses finite blocks of length $R$, the result should be interpreted as a block-level stationarity guarantee rather than convergence of the full periodically refreshed sequence. For each finite block, the proof yields descent and Zoutendijk summability over that block. 
\end{theorem}
\begin{proof}
By Lemma~\ref{lemma:2}, the gradient is Lipschitz continuous on $\mathcal L$.
By Lemma~\ref{lemma:3}, the directions satisfy sufficient descent. Therefore, the strong Wolfe conditions imply the Zoutendijk condition \cite{no2006num,gilbert1992glob}
\begin{align}
\sum_{k=0}^{\infty}
\frac{(\mathbf g_k^\top \mathbf p_k)^2}{\|\mathbf p_k\|_2^2}
<
\infty.
\end{align} 
Using Lemma~\ref{lemma:3} and Assumption~\ref{assum:6},
\begin{align}
\frac{(\mathbf g_k^\top \mathbf p_k)^2}{\|\mathbf p_k\|_2^2}
\ge
\frac{\eta^2\|\mathbf g_k\|_2^4}{M_p^2\|\mathbf g_k\|_2^2}
=
\frac{\eta^2}{M_p^2}\|\mathbf g_k\|_2^2 .
\end{align}
Hence $\sum_{k=0}^{\infty}\|\mathbf g_k\|_2^2<\infty$, which implies $\|\mathbf g_k\|_2\to0$.
\end{proof}
The analysis above establishes stationarity only for a fixed smooth objective optimized within a block. It is not a proof of convergence for the entire periodically refreshed sequence, because updating $\widehat F_i$, $\widehat f_i$, or $\Sigma$ changes the objective. Proving convergence of the full nonstationary procedure would require additional assumptions controlling the size and frequency of these estimator updates.

\section{Numerical Experiments}
\label{sec:experiments}
This section evaluates the proposed CIC-CG algorithm under controlled synthetic multivariate regression  conditions. The experiments are designed to isolate and assess three key aspects: (i) robustness to heavy-tailed noise, (ii) the benefit of explicitly modeling dependence among multivariate errors, and (iii) convergence and tail-behavior relative to established robust learning criteria. All methods are implemented using identical model architectures, training protocols, and stopping criteria to ensure a fair and meaningful comparison. Unless otherwise stated, all reported results are averaged over multiple Monte Carlo trials.

\subsection{Experimental Setup}
\label{subsec:setup}
\subsubsection{Model Architecture}
We consider a nonlinear multivariate regression problem implemented using a multilayer perceptron (MLP) with a single hidden layer. The MLP consists of $d=3$ input nodes, $p=3$ output nodes, and a fixed hidden layer of $14$ neurons with hyperbolic tangent activation functions. This architecture is sufficiently expressive to capture nonlinear input-output relationships while remaining standard in signal processing and adaptive learning literature. All methods use the same network architecture and initialization strategy. Training is performed for a fixed number of iterations using full-batch gradients, so that performance differences are attributable solely to the learning criterion.

\subsubsection{Implementation of the CIC Objective}
The numerical experiments use the mixed objective in \eqref{eq:CIC_objective}, so the optimized training criterion is consistent with the derivations in Sections \ref{sec:citl}-\ref{sec:conv}.
The marginal term $\psi_{\rm marg}$ provides classical correntropy robustness in the raw residual domain, whereas $\psi_{\rm dep}$ penalizes deviations in copula-transformed residual space.
The parameter $\gamma$ controls the relative contribution of dependence modeling: $\gamma=0$ reduces the objective to a purely marginal correntropy criterion without explicit dependence modeling, whereas $\gamma>0$ progressively emphasizes the copula-space dependence term. Thus, the $\gamma=0$ variant does not include the copula-space dependence penalty; any difference from classical MCC must arise from the multivariate marginal aggregation and the common optimization protocol rather than from explicit dependence modeling.
Unless otherwise stated, $\gamma=0.55$ is used in all experiments. Table~\ref{tab:hyperparameters} summarizes the experimental parameters that should be fixed across Monte Carlo runs for reproducibility.

\begin{table}[t]
\caption{Experimental Hyperparameters and Simulation Settings}
\label{tab:hyperparameters}
\centering
\footnotesize
\begin{tabular}{ll}
\hline
Parameter & Value used in experiments \\ \hline
CIC shape parameter $\alpha$ & $1.2$ \\ 
Regularization $\delta$ & $10^{-12}$ \\ 
Mixing weight $\gamma$ & $0.55$ \\ 
Kernel width $\sigma_{\rm k}$ & $0.8$ \\ 
Noise scale $\sigma_{\varepsilon}$ & $0.35$ in the hard regime \\ 
Shrinkage parameter $\lambda$ & $0.35$ \\ 
Covariance ridge $\varepsilon_\Sigma$ & $10^{-8}$\\ 
Marginal transform & Student's-$t$ CDF, DoF $=1$, scale $=1$ \\ 
Estimator update period $R$ & $15$ \\ 
Restart threshold $\eta$ & $10^{-3}$ \\ 
Bounded-direction constant $M_p$ & $10$ \\ 
Number of iterations & $70$ \\ 
Monte Carlo runs & $50$ \\ 
\hline
\end{tabular}
\end{table}

\subsubsection{Synthetic Data Generation}
Input samples $\mathbf{x}_n \in \mathbb{R}^3$ are drawn independently from a uniform distribution on $[-1,1]^3$. The clean multivariate target signal $\mathbf{d}_n \in \mathbb{R}^3$ is generated according to a synthetic nonlinear benchmark mapping designed to induce nonlinear interactions among inputs \cite{nare1990iden}, i.e.,
\begin{align}
	\mathbf d_n =
	\begin{bmatrix}
		\sin(1.2\,x_{n,1}x_{n,2}) + 0.2\cos(2x_{n,3}) \\
		\cos(1.0\,x_{n,2}x_{n,3}) + 0.3\sin(1.5x_{n,1}) \\
		x_{n,1}^2 - x_{n,3} + 0.15\sin(2.2x_{n,2})
	\end{bmatrix}.
\end{align}
The training set consists of $N_{\mathrm{tr}}=600$ samples,
and the test set consists of $N_{\mathrm{te}}=600$ samples.

\subsubsection{Noise Model (Heavy-Tailed and Dependent Disturbances)}

To rigorously evaluate robustness and dependence modeling,
additive noise is injected into the training data using a multivariate Student's-$t$ distribution with copula-induced dependence. Specifically, the noise vector $\boldsymbol{\varepsilon}_n \in \mathbb{R}^p$ is generated as
$
	\boldsymbol{\varepsilon}_n = \sigma_{\varepsilon} \, \mathbf{T}_n
$,
where $\mathbf{T}_n$ follows a $p$-dimensional Student's-$t$ distribution with DoF $\nu$ and a correlation structure induced by a Student's-$t$ copula parameterized by dependence strength $\rho$. The dependence structure is specified through the constant correlation matrix $\mathbf R=(1-\rho)\mathbf I+\rho\,\mathbf 1\mathbf 1^\top$, with $0\le\rho<1$ in the simulations to preserve positive definiteness; numerical sweeps approaching $\rho=1$ use $\rho=0.99$ or an equivalent small jitter. Thus, all noise components share pairwise correlation coefficient $\rho$. This correlation matrix is used within the Student's-$t$ copula to generate dependent heavy-tailed disturbances. Additionally, by varying $\nu$, we control the heaviness of the marginal tails, while $\rho$ governs the strength of cross-component dependence and tail co-occurrence. This construction enables a systematic evaluation of both marginal robustness
and dependence-awareness. Unless otherwise stated, the ``hard" regime corresponds to $\nu=2.2$, $\sigma_{\varepsilon}=0.35$, and $\rho=0.85$, representing highly heavy-tailed and highly dependent noise. Importantly, all performance metrics are evaluated with respect to \emph{clean} test targets, ensuring that methods are assessed based on generalization
rather than fitting noisy labels.

\subsubsection{Benchmarked Methods}
The proposed CIC-CG method is compared against several established learning criteria. As a benchmark, we consider classical MSE learning, which is optimal under independent Gaussian noise assumptions. To evaluate marginal robustness, we include the Huber loss \cite{huber92rob} and the Student's-$t$ negative log-likelihood \cite{lang1989rob},  which model heavy-tailed marginals under independence assumptions.  We also consider the classical maximum correntropy criterion (MCC) with Gaussian kernel  \cite{liu2011kernel}, which provides bounded influence and robustness to impulsive noise. Finally, the proposed CIC-CG method extends correntropy by incorporating copula-based dependence modeling.
All objectives are optimized using the same full-batch PRP$^+$ nonlinear-CG framework and the same initialization and stopping conditions, unless explicitly stated otherwise. This makes the comparison primarily attributable to the objective function and its dependence-modeling terms rather than to optimizer-specific confounders.

\subsubsection{Evaluation Metrics}
Performance is evaluated using metrics that jointly assess overall prediction accuracy and robustness to extreme deviations.  Specifically, we report the root mean squared error (RMSE) computed on clean test data, which measures average generalization performance.  To characterize robustness under heavy-tailed disturbances, we also consider upper quantiles of the absolute error, namely Q90 and Q95, which emphasize tail behavior and rare but significant deviations.  In addition, convergence properties are analyzed by tracking the evolution of these test-set metrics as a function of training iteration.

These metrics jointly reveal how different learning criteria behave under heavy-tailed and dependent noise, particularly in the presence of rare but severe outliers. In addition, the proposed CIC-CG method is implemented using a generalized objective that balances marginal robustness and dependence-aware learning through the same convex mixing parameter $\gamma \in [0,1]$.

\subsubsection{Monte Carlo Protocol}
All reported results are averaged over $50$ independent Monte Carlo runs. In each run, the network parameters are randomly initialized and new noise realizations are generated for the training targets only; all reported test metrics are computed against clean test targets. Averaging across runs reduces variance and ensures statistical reliability of the observed trends. In the final presentation, mean curves should be accompanied by standard deviations, confidence intervals, or shaded error bands to quantify Monte Carlo variability; alternatively, a table of mean $\pm$ standard deviation for RMSE, Q90, and Q95 can be included. 
 \begin{figure}[!t]
	\centering
	\subfloat[]{
		\includegraphics[width=0.8\linewidth]{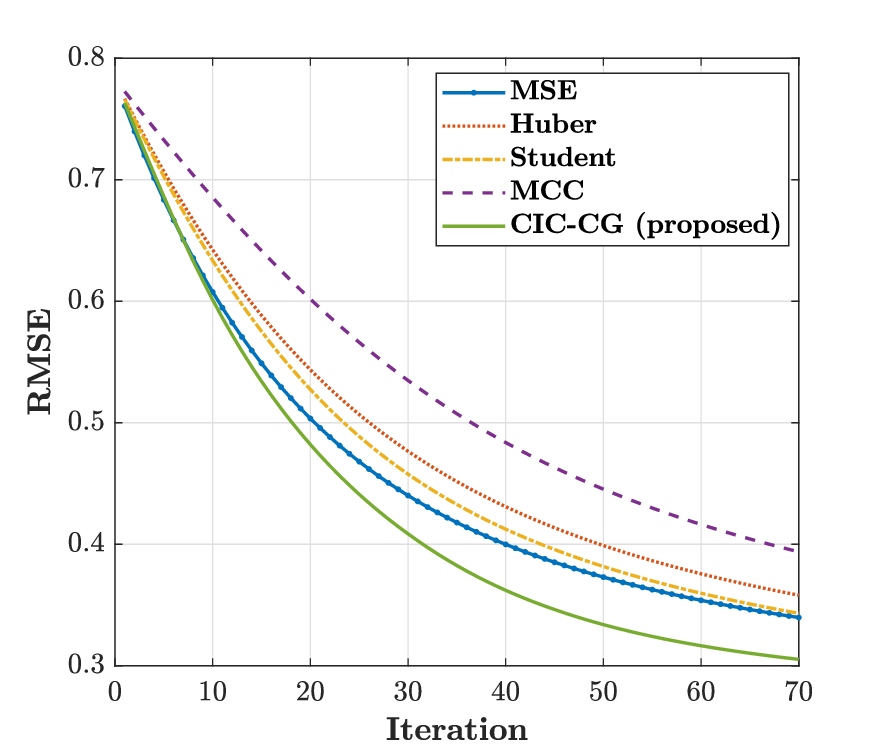}
		\label{fig:rmse-it}
	}\\
	\vspace{1mm}
	\subfloat[]{
		\includegraphics[width=0.8\linewidth]{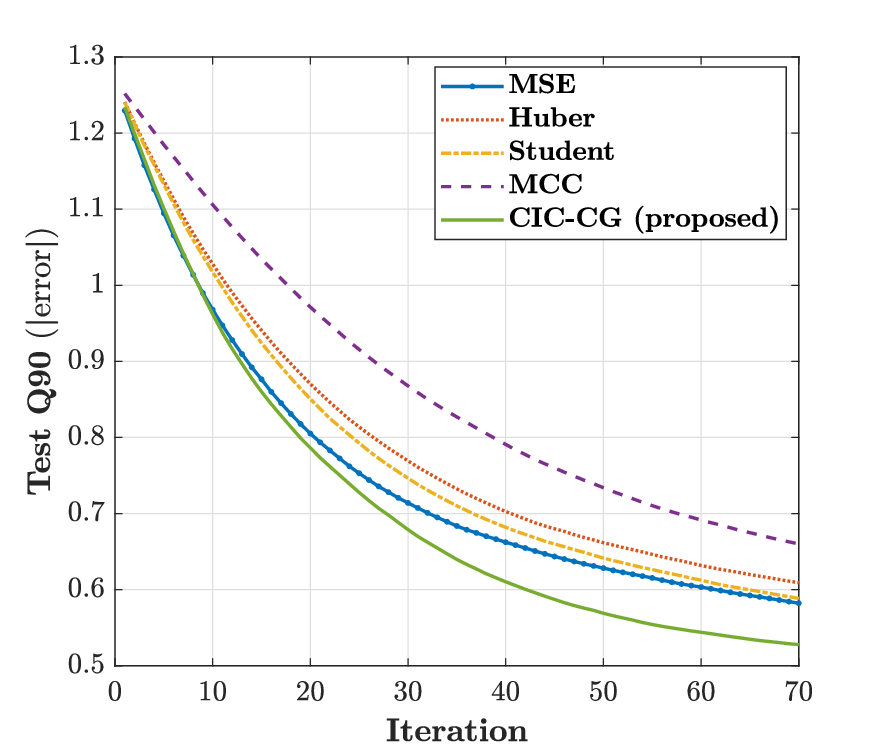}
		\label{fig:q90-it}
	}
	\caption{Convergence behavior under dependent heavy-tailed noise. (a) Test RMSE versus iteration; and (b) Test 90\% quantile of absolute error (Q90) versus iteration.
	}
	\label{fig:main1}
\end{figure}

\begin{figure}[!t]
	\centering
	\subfloat[]{
		\includegraphics[width=0.8\linewidth]{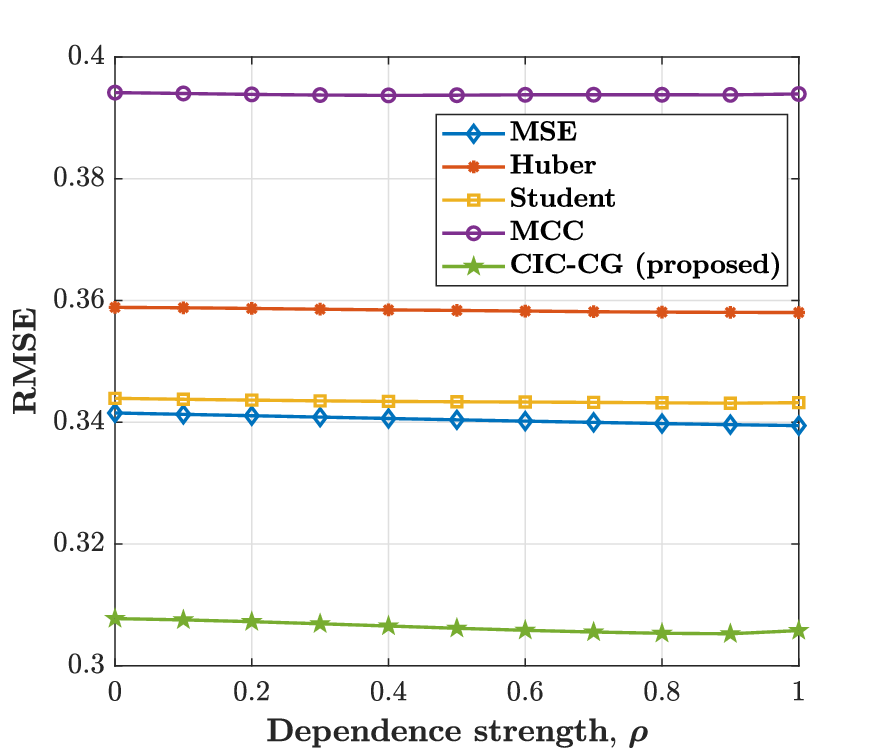}
		\label{fig:rmse-dep}
	}\\
	\vspace{1mm}
	\subfloat[]{
		\includegraphics[width=0.8\linewidth]{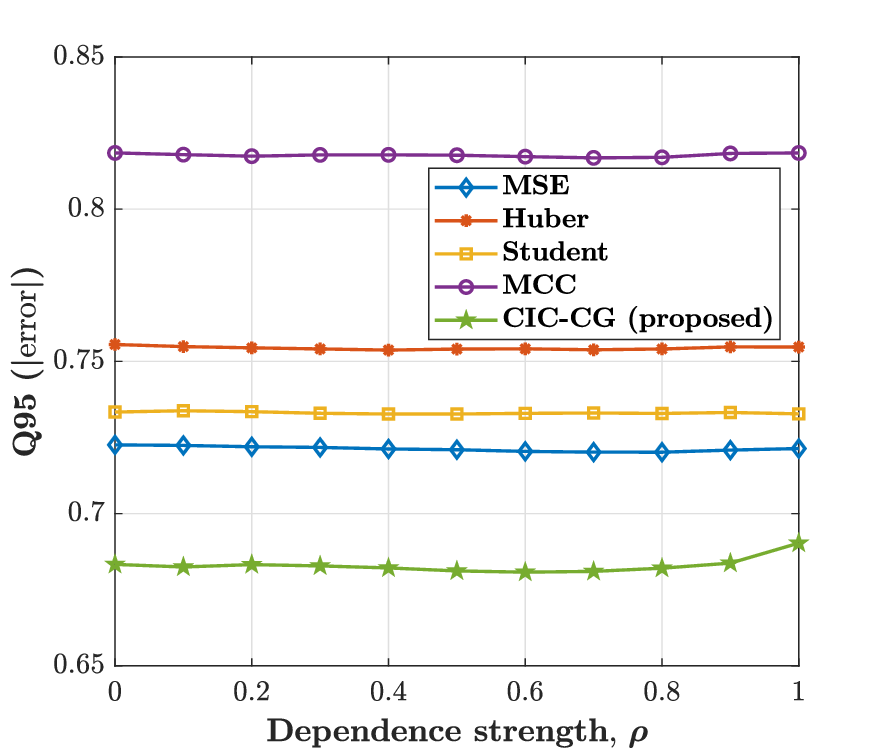}
		\label{fig:q95-dep}
	}
	\caption{Impact of dependence strength on performance under heavy-tailed noise. (a) Test RMSE versus dependence strength $\rho$; and (b) Test 95\% quantile of absolute error (Q95) versus dependence strength $\rho$.
	}
	\label{fig:main2}
\end{figure}
%

\begin{figure}[t]
	\centering
	\includegraphics[width=1\columnwidth]{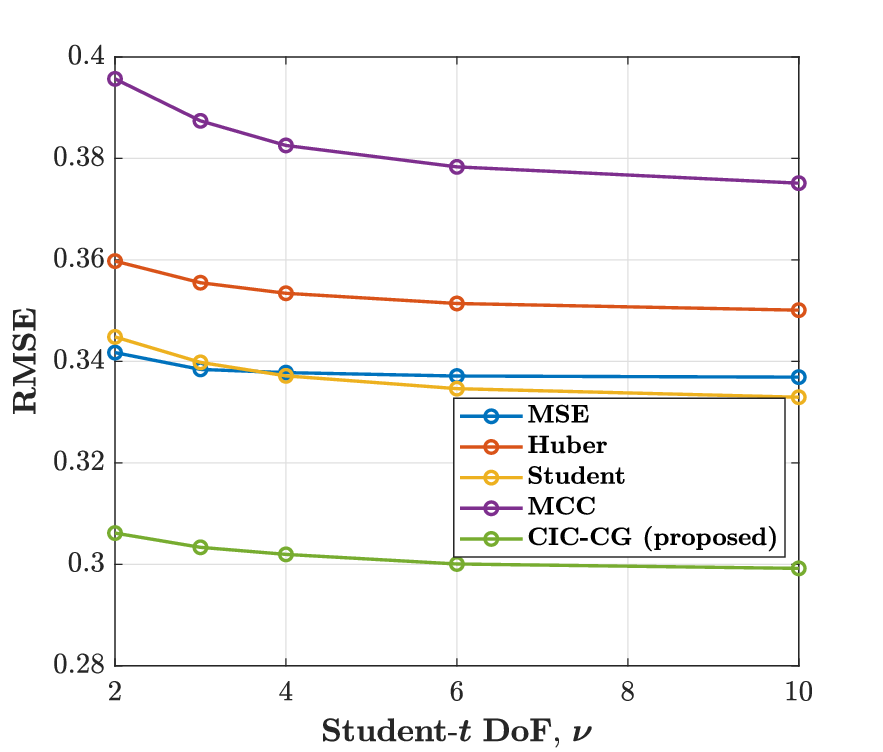}
	\caption{Robustness to marginal impulsiveness: test RMSE versus Student's-$t$ DoF $\nu$ under dependent heavy-tailed noise.}\label{fig:rmse-nu}
\end{figure}

\begin{figure}[t]
	\centering
	\includegraphics[width=1\columnwidth]{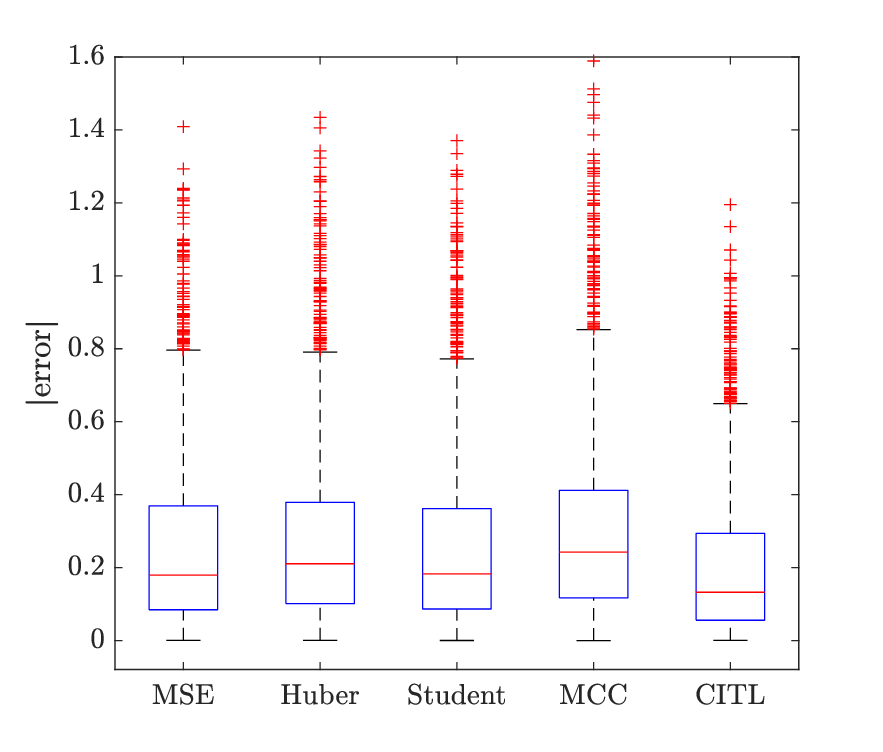}
	\caption{Distribution of absolute prediction errors $|d_{\mathrm{clean}}-\hat y|$ under a hard dependent heavy-tailed noise regime (clean evaluation). The label ``CITL" in the plotted legend denotes the proposed CIC-CG implementation. 
	}\label{fig:err-all}
\end{figure}

\begin{figure}[t]
	\centering
	\includegraphics[width=1\columnwidth]{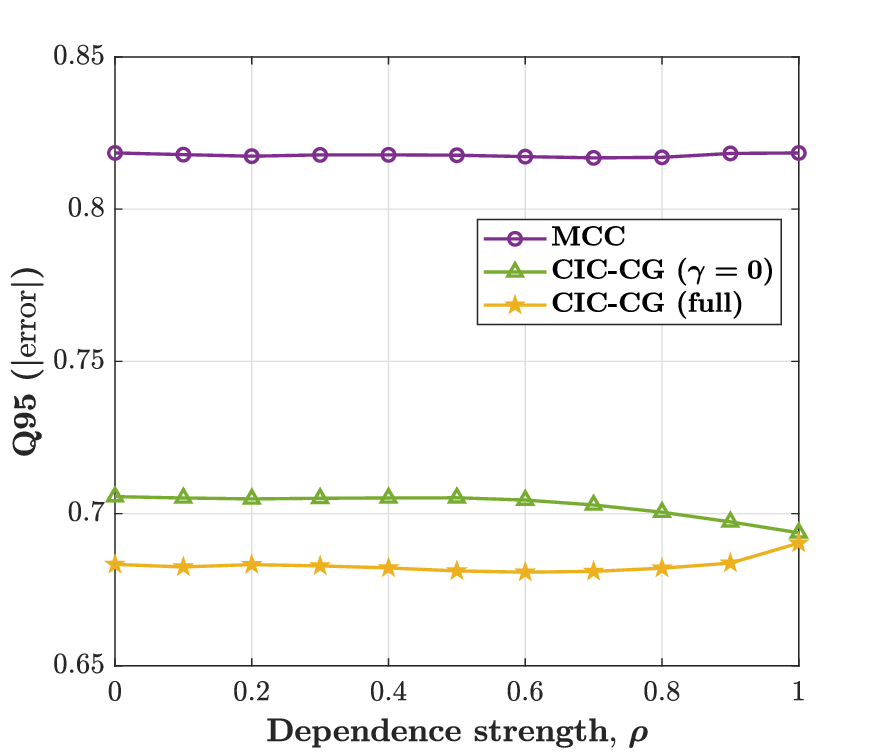}
	\caption{Ablation study on tail robustness under dependent heavy-tailed noise.
		Test $Q_{95}(|e|)$ is reported versus the dependence strength $\rho$, evaluated with respect to clean targets. Results compare classical MCC, CIC without dependence modeling ($\gamma=0$), and the full proposed CIC-CG with the joint marginal--dependence objective.
	}\label{fig:q95-rho}
\end{figure}

Fig.~\ref{fig:main1} illustrates the convergence behavior of different learning criteria under dependent heavy-tailed noise, evaluated with respect to the clean test targets. Fig.~\ref{fig:main1}(a) reports the evolution of the test RMSE, while Fig.~\ref{fig:main1}(b) shows the test 90\% quantile of the absolute error (Q90), which emphasizes tail robustness. It is evident that the proposed CIC-CG method consistently converges to the lowest error level among all compared approaches in both metrics, which indicates that incorporating dependence information via copula-induced correntropy leads to improved generalization accuracy as well as superior suppression of large deviations. Additionally, it can be observed that, while classical robust losses such as Huber and Student's-$t$ improve over MSE by mitigating the influence of marginal outliers, their convergence curves remain systematically above that of CIC-CG. 
This behavior is expected, as these methods treat error components independently and therefore do not explicitly weight joint deviations arising from dependent noise, including multivariate outliers that depart from the dominant dependence directions. Moreover, the figures reveal that the classical MCC exhibits slower convergence and higher steady-state error compared to CIC-CG. Although correntropy effectively limits the influence of large errors, its isotropic and componentwise formulation does not account for dependence structure, which becomes critical under copula-induced heavy-tailed disturbances.

Comparing Fig.~\ref{fig:main1}(a) and Fig.~\ref{fig:main1}(b), the relative performance trends are consistent.  Methods that achieve lower RMSE also exhibit improved tail behavior. 
However, the performance gap between CIC-CG and the baselines is slightly more pronounced in the Q90 metric, highlighting the particular strength of the proposed approach in controlling extreme errors. This confirms that the advantage of CIC-CG is not merely due to improved average fitting, but rather stems from its ability to jointly model marginal robustness and dependence-aware tail behavior.

Fig.~\ref{fig:main2} examines the effect of dependence strength on learning performance under heavy-tailed noise. Fig.~\ref{fig:main2}(a) reports the test RMSE as a function of the copula dependence parameter $\rho$, while Fig.~\ref{fig:main2}(b) shows the corresponding tail metric based on the 95\% quantile of the absolute error. A key observation is that the proposed CIC-CG method consistently achieves the lowest error across the entire range of dependence strengths in both metrics. In particular, the performance advantage of CIC-CG is preserved as $\rho$ increases,
indicating that the method remains stable even when strong dependence is present among error components. 
  In contrast, baseline methods that do not explicitly model dependence  exhibit largely flat performance as $\rho$ increases, with only minor  variations across the range of dependence strengths. This behavior is expected, as losses such as MSE, Huber, Student's-$t$, and classical MCC implicitly assume independent or weakly dependent errors and therefore do not explicitly adapt to changes in the joint dependence structure.

The gap between CIC-CG and the baselines is more pronounced in Fig.~\ref{fig:main2}(b), which emphasizes extreme errors. This highlights an important property of the proposed framework: by operating in copula-transformed space, CIC-CG is able to suppress joint deviations that are not explicitly weighted by marginally robust criteria.  As a result, increases in dependence strength do not lead to a noticeable increase in tail risk for CIC-CG, while the baseline methods remain largely insensitive to $\rho$ due to their reliance on marginal error statistics. It is also worth noting that the RMSE curves in Fig.~\ref{fig:main2}(a) remain relatively stable for all methods, suggesting that dependence primarily affects higher-order error statistics rather than average performance. This further motivates the use of tail-sensitive metrics when evaluating robust learning under dependent noise. However, the benchmark methods operate componentwise on residuals and depend primarily on marginal error statistics. Since increasing $\rho$ modifies the dependence structure without altering the marginal distributions, their tail performance remains relatively stable even as the dependence strength increases.

Fig. \ref{fig:rmse-nu} investigates the robustness of different learning criteria with respect to marginal impulsiveness, controlled by the DoF  $\nu$ of the Student's-$t$ noise. Smaller values of $\nu$ correspond to heavier tails and more frequent extreme outliers, \textcolor{black}{whereas larger $\nu$ leads the Student's-$t$ marginals to approach the Gaussian distribution (in the limit $\nu \to \infty$), thereby reducing marginal tail heaviness while the dependence structure remains governed by the underlying copula.} As expected, all methods benefit from increasing $\nu$, as the noise becomes less impulsive and easier to handle. However, the rate and extent of improvement differ substantially across methods. The MSE and Huber losses exhibit limited robustness for small $\nu$, reflecting their sensitivity to heavy-tailed disturbances. While Huber improves upon MSE by clipping large residuals, it remains fundamentally marginal and fails to adapt to strong outliers. The Student's-$t$ loss provides enhanced robustness by explicitly modeling heavy-tailed marginals, resulting in lower RMSE across all  $\nu$ values. Nevertheless, its performance saturates as  $\nu$ increases, indicating that marginal robustness alone is insufficient in the presence of structured noise. Classical MCC further improves robustness by suppressing large deviations through kernel weighting. However, MCC operates independently across output dimensions and implicitly assumes residual independence, which limits its effectiveness when noise components are statistically dependent.

In contrast, the proposed CIC-CG consistently achieves the lowest RMSE across the entire range of $\nu$. The performance gap is most pronounced in the strongly impulsive regime (small $\nu$), where tail events are both frequent and dependent. This behavior directly reflects the core advantage of the proposed framework: marginal robustness inherited from correntropy combined with explicit dependence modeling via the copula transformation. Importantly, as $\nu$ increases and the noise becomes less heavy-tailed, the advantage of CIC-CG gradually narrows but does not vanish. This indicates that the copula-induced dependence modeling does not harm performance in near-Gaussian conditions, while providing substantial gains in challenging heavy-tailed regimes.

Fig.~\ref{fig:err-all} presents the distribution of absolute prediction errors $|d_{\mathrm{clean}}-\hat y|$ under a challenging dependent heavy-tailed noise regime. All models are trained using noisy labels, while evaluation is performed against the clean test targets in order to assess true generalization rather than noise fitting. The boxplot representation highlights not only the central tendency of the errors but also their dispersion and tail behavior, which are critical in heavy-tailed environments. As observed, MSE and Huber losses exhibit wide interquartile ranges and long tails, indicating strong sensitivity to extreme outliers. The Student's-$t$ loss reduces tail magnitude by explicitly modeling heavy-tailed marginals, yet still produces a considerable number of large errors due to the lack of dependence modeling.  Classical MCC applies kernel-based weighting to reduce the influence of large errors.  However, in this dependent noise setting its componentwise formulation limits its  effectiveness, resulting in a wider dispersion of errors and a larger number of  extreme deviations compared to the proposed method. 
In contrast, the proposed CIC-CG method achieves the smallest median error, the narrowest interquartile range, and a substantially reduced number of extreme outliers. This behavior directly reflects the advantage of incorporating copula-induced dependence into the correntropy framework, enabling effective suppression of dependent tail events. Therefore, the proposed CIC-CG not only improves average accuracy but also significantly reduces tail risk, which is essential for robust signal processing
under dependent heavy-tailed disturbances.

Fig.~\ref{fig:q95-rho} investigates the role of the dependence term in the proposed copula-induced correntropy learning criterion through an ablation study. We compare three methods: classical MCC, CIC-CG with $\gamma=0$ (marginal correntropy only), and the full CIC-CG with joint marginal--dependence modeling. It can be observed that the MCC baseline exhibits nearly constant performance as $\rho$ increases, indicating that classical correntropy is insensitive to cross-component dependence and tail co-movements. While MCC is robust to marginal outliers, it lacks a mechanism to explicitly exploit or model cross-component dependence. As a result, its performance remains largely insensitive to changes in $\rho$, since the marginal error statistics remain unchanged. The $\gamma=0$ variant should be interpreted as the marginal component of the mixed objective, not as a copula-dependence model. Its difference from classical MCC is due to the multivariate marginal aggregation and the common PRP$^+$ optimization protocol used in this study.
The full CIC-CG achieves the lowest $Q_{95}(|e|)$ over a wide range of dependence strengths, demonstrating the benefit of adding the copula space dependence penalty to marginal correntropy. A mild increase in error at very high $\rho$ may reflect a bias--variance trade-off introduced by a fixed mixing weight $\gamma$, motivating adaptive or data-driven tuning of $\gamma$.

\vspace{-0.4cm}
\section{Conclusion}
\label{sec:conclusion}

This paper proposed a copula-induced correntropy learning framework for robust signal processing in the presence of dependent heavy-tailed noise. By embedding a copula-transformed residual representation into an information-theoretic similarity measure, the resulting CITL criterion separates marginal robustness from inter-component dependence-aware weighting, addressing a key limitation of   componentwise correntropy-based approaches. We developed a nonlinear conjugate-gradient learning algorithm for the mixed marginal--dependence objective used in the experiments. For fixed marginal estimators, a fixed copula-space metric, and a regularized radial penalty, we established sufficient descent and global stationarity for the fixed-estimator subproblem under standard Wolfe line-search and bounded-direction assumptions. Numerical experiments on synthetic multivariate regression tasks demonstrated that the proposed CIC-CG method improves robustness under dependent impulsive noise, particularly in terms of tail-sensitive error metrics. The ablation study further indicated that the copula-space dependence term provides gains beyond marginal robust losses alone, although the current covariance copula-space metric should be viewed as a tractable dependence weighting mechanism rather than a full copula density estimator; richer copula density models would be required to characterize nonlinear or asymmetric tail dependence completely.
\vspace{-0.3cm}
\bibliographystyle{IEEEtran}

\begin{thebibliography}{99}


	\bibitem{maronna2019robust}
R. A. Maronna, R. D. Martin, V. J. Yohai and M. Salibi\'an-Barrera, \emph{Robust Statistics: Theory and Methods (with R)}, {John Wiley \& Sons}, 2019.

	\bibitem{hay2002adaptive}
S. S. Haykin, \emph{Adaptive Filter Theory}, {Pearson Education India}, 2002.

\bibitem{sun2017major}
Y. Sun, P. Babu, and D. P. Palomar, ``Majorization-Minimization Algorithms in Signal Processing, Communications, and Machine Learning," \emph{IEEE Trans. Signal Process.}, vol. 65, no. 3, pp. 794-816, Feb. 2017.

	\bibitem{huber2011robust}
P. J. Huber, ``Robust statistics," \emph{ International encyclopedia of statistical science}, Springer, Berlin, Heidelberg, pp. 1248-1251, 2011.

\bibitem{jav2025time}
A. Javaheri, J. Ying, D. P. Palomar and F. Marvasti, ``Time-Varying Graph Learning for Data With Heavy-Tailed Distribution," \emph{IEEE Trans. Signal Process.}, vol. 73, pp. 3044-3060, 2025.

	\bibitem{clavier2021imp}
L. Clavier, G. W. Peters, F. Septier and I. Nevat, ``Impulsive Noise Modeling and Robust Receiver Design," \emph{EURASIP J. Wirel. Commun. Netw.}, vol. 13, no. 1, 2021.


\bibitem{nasri2009ad}
A. Nasri, A. Nezampour and R. Schober, ``Adaptive Lp-norm Diversity Combining in Non-Gaussian Noise and Interference," \emph{IEEE Trans. Wireless Commun.}, vol. 8, no. 8, pp. 4230-4240, Aug. 2009.

\bibitem{zhu2018huber}
B. Zhu, L. Chang, J. Xu, F. Zha and J. Li, ``Huber-based Adaptive Unscented Kalman Filter with Non-Gaussian Measurement Noise," \emph{Circuits Syst. Signal Process.}, vol. 37, no. 9, pp. 3842-3861, 2018.

\bibitem{tang2024gen}
H. Tang, H. Han, S. Zhang and W. Feng, ``A Generalized t-Distribution-Based Kernel Adaptive Filtering Algorithm," \emph{IEEE Trans. Circuits Syst. {II}: Express Briefs}, vol. 71, no. 6, pp. 3241-3245, June 2024.

\bibitem{huang2019novel}
Y. Huang, Y. Zhang, Y. Zhao and J. A. Chambers, ``A Novel Robust Gaussian-Student's t Mixture Distribution Based Kalman Filter," \emph{IEEE Trans. Signal Process.}, vol. 67, no. 13, pp. 3606-3620, July  2019.

	\bibitem{gr2022an}
J. Größer and O. Okhrin, ``Copulae: An Overview and Recent Developments," \emph{Wiley Interdiscip. Rev. Comput. Stat.}, vol. 14, no. 3, 2022.


	\bibitem{prin2010info}
J. C. Principe, \emph{Information theoretic learning: Renyi's entropy and kernel perspectives}, {Springer Science \& Business Media}, 2010.

\bibitem{liu2007corr}
W. Liu, P. P. Pokharel and J. C. Principe, ``Correntropy: Properties and Applications in Non-Gaussian Signal Processing," \emph{IEEE Trans. Signal Process.}, vol. 55, no. 11, pp. 5286-5298, 2007.

\bibitem{heravi2018new}
A. R. Heravi and G. Abed Hodtani, ``A New Correntropy-Based Conjugate Gradient Backpropagation Algorithm for Improving Training in Neural Networks," \emph{IEEE Trans. Neural Netw. Learn. Syst.}, vol. 29, no. 12, pp. 6252-6263, Dec. 2018.

\bibitem{chen2017rob}
B. Chen, L. Xing, H. Zhao, B. Xu and J. C. Principe, ``Robustness of maximum correntropy estimation against large outliers," \emph{arXiv preprint}, \url{https://arxiv.org/abs/1703.08065v2}, 2017. 

\bibitem{chen2016gen}
B. Chen, L. Xing, H. Zhao, N. Zheng and J. C. Pr\'incipe, ``Generalized Correntropy for Robust Adaptive Filtering," \emph{IEEE Trans. Signal Process.}, vol. 64, no. 13, pp. 3376-3387, July, 2016.

\bibitem{nelson2006an}
R. B. Nelsen, \emph{An Introduction to Copulas}, {New York, NY: Springer New York}, Jan. 2006. 


\bibitem{zeng2014copula}
X. Zeng, J. Ren, Z. Wang, S. Marshall and T. Durrani, ``Copulas for Statistical Signal Processing (Part I): Extensions and Generalization," \emph{Signal Process.}, vol. 94, pp. 691-702, 2014.

\bibitem{peter2014comm}
G. W. Peters, T. A. Myrvoll, T. Matsui, I. Nevat and F. Septier, ``Communications meets copula modeling: Non-standard dependence features in wireless fading channels," \emph{2014 IEEE Global Conference on Signal and Information Processing (GlobalSIP)}, Atlanta, GA, USA, 2014, pp. 1224-1228. 

\bibitem{ghadi2021copula}
F. Rostami Ghadi and G. A. Hodtani, ``Copula-Based Analysis of Physical Layer Security Performances Over Correlated Rayleigh Fading Channels," \emph{IEEE Transactions on Information Forensics and Security}, vol. 16, pp. 431-440, 2021.

\bibitem{zheng2019copula}
C. Zheng, M. Egan, L. Clavier, G. W. Peters and J. -M. Gorce, ``Copula-Based Interference Models for IoT Wireless Networks," \emph{2019 IEEE International Conference on Communications (ICC)}, Shanghai, China, 2019, pp. 1-6.

\bibitem{jor2021copula}
E. A. Jorswieck and K. -L. Besser, ``Copula-Based Bounds for Multi-User Communications–Part I: Average Performance," \emph{IEEE Communications Letters}, vol. 25, no. 1, pp. 3-7, Jan. 2021.


\bibitem{chen2012max}
B. Chen and J. C. Principe, ``Maximum Correntropy Estimation Is a Smoothed MAP Estimation," \emph{IEEE Signal Process. Lett.}, vol. 19, no. 8, pp. 491-494, Aug. 2012.

\bibitem{no2006num}
J. Nocedal and S. J. Wright, \emph{Numerical optimization}, {New York, NY: Springer New York}, Jul. 2006.

\bibitem{chat2010fle}
A. Chatterjee, ``A Fletcher-Reeves Conjugate Gradient Neural-Network-Based Localization Algorithm for Wireless Sensor Networks," \emph{IEEE Trans. Veh. Technol.}, vol. 59, no. 2, pp. 823-830, Feb. 2010.

\bibitem{zhang2006des}
L. Zhang, W. Zhou and D. H. Li, ``A Descent Modified Polak-Ribi\'ere-Polyak Conjugate Gradient Method and Its Global Convergence," \emph{IMA J. Numer. Anal.}, vol. 24, no. 6, pp. 629-640, 2006.

\bibitem{noc2005new}
W. W. Hager and H. Zhang, ``A New Conjugate Gradient Method with Guaranteed Descent and an Efficient Line Search," \emph{SIAM J. Optim.}, vol. 16, no. 1, pp. 170-192, 2005.

\bibitem{silv2018dens}
B. W. Silverman, \emph{Density estimation for statistics and data analysis}, {Routledge}, 2018.

\bibitem{leo2004well}
O. Ledoit and M. Wolf, ``A Well-conditioned Estimator for Large-Dimensional Covariance Matrices," \emph{J. Multivar. Anal.}, vol. 88, no. 2, pp. 365-411, 2004.

\bibitem{chen2010sh}
Y. Chen, A. Wiesel, Y. C. Eldar and A. O. Hero, ``Shrinkage Algorithms for MMSE Covariance Estimation," \emph{IEEE Trans. Signal Process.}, vol. 58, no. 10, pp. 5016-5029, Oct. 2010.

\bibitem{gilbert1992glob}
J. C. Gilbert and J. Nocedal, ``Global Convergence Properties of Conjugate Gradient Methods for Optimization," \emph{SIAM J. Optim.}, vol. 2, no. 1, pp. 21-42, 1992.

\bibitem{li2023new}
Y. Li, C. Li, W. Yang and W. Zhang, ``A New Conjugate Gradient Method with A Restart Direction and Its Application in Image Restoration," \emph{AIMS Math}, vol. 8, no. 12, pp. 28791-28807, 2023.
\bibitem{huber92rob}
P. J. Huber, ``Robust estimation of a location parameter," \emph{Breakthroughs in statistics: Methodology and distribution}, pp. 492-518, New York, NY: Springer New York, 1992.


\bibitem{lang1989rob}
K. L. Lange, R. J. A. Little and J. M. G. Taylor, ``Robust Statistical Modeling Using the $t$ Distribution,"
 \emph{J. Am. Stat. Assoc.}, vol. 84, no. 408, pp. 881-896, 1989.

 
 \bibitem{liu2011kernel}
 W. Liu, J. C. Principe and S. Haykin, \emph{Kernel Adaptive Filtering: A Comprehensive Introduction}, John Wiley \& Sons, 2011.

 
\bibitem{nare1990iden}
K. S. Narendra and K. Parthasarathy, ``Identification and control of dynamical systems using neural networks," \emph{IEEE Trans. Neural Netw.}, vol. 1, no. 1, pp. 4-27, March 1990.
	\end{thebibliography}

\end{document}